\documentclass[a4paper,UKenglish,usenames]{lipics-v2021}
\usepackage{amsthm}
\usepackage{complexity}
\usepackage{thmtools}

\pdfoutput=1 
\hideLIPIcs  

\hypersetup{
	colorlinks=true,
	linkcolor=blue,
	filecolor=magenta,
	citecolor = black,      
	urlcolor=cyan,
}
\newcommand{\mcU}{\ensuremath{\mathcal{U}}}
\newcommand{\mcV}{\ensuremath{\mathcal{V}}}

\newcommand{\mcA}{\ensuremath{\mathcal{A}}}
\newcommand{\mcB}{\ensuremath{\mathcal{B}}}
\newcommand{\mcI}{\ensuremath{\mathcal{I}}}

\newcommand{\apr}[1]{\ensuremath{\tilde{F}_{#1}}}

\newcommand{\strings}{\ensuremath{\{0,1\}^d}}

\newcommand{\floor}[1]{\ensuremath{\left\lfloor #1 \right\rfloor}}
\newcommand{\Hd}{\ensuremath{\mathcal{H}_{d/2}}}
\newcommand{\var}{\ensuremath{\texttt{var}}}

\newcommand{\OR}{\ensuremath{\mathsf{OR}}}
\newcommand{\AND}{\ensuremath{\mathsf{AND}}}

\newcommand{\ORANDOR}{\ensuremath{\OR\circ\AND\circ\OR}}

\newcommand{\kov}[3]{\ensuremath{\mathsf{#1OV}_{#2,{#3}}}} 
\newcommand{\kovnd}{\kov{k}{n}{d}} 

\newcommand{\kayov}{\ensuremath{\mathsf{kOV}}}
\newcommand{\OVnd}{\mathsf{OV}_{n,d}}
\newcommand{\OV}{\ensuremath{\mathsf{OV}}}

\newcommand{\kint}[3]{\ensuremath{\mathsf{#1Int}_{#2,{#3}}}} 
\newcommand{\kintnd}{\kint{k}{n}{d}} 
\newcommand{\twoint}{\kint{2}{n}{d}}

\newcommand{\clique}{\ensuremath{\mathsf{Clique}}}
\newcommand{\hypcliquemrl}{\ensuremath{\mathsf{HypClique}_{r,\ell}^m}}
\newcommand{\Intnd}{\ensuremath{\mathsf{Int}_{n,d}}}
\newcommand{\Int}{\ensuremath{\mathsf{Int}}}

\newcommand{\Dz}{\mathbf{D_0}}
\newcommand{\Do}{\mathbf{D_1}}

\newcommand{\customComment}[2]{}
\newcommand{\karteek}[1]{\customComment{K}{#1}}
\newcommand{\nutan}[1]{\customComment{N}{#1}}
\newcommand{\srikanth}[1]{\customComment{S}{#1}}
\newcommand{\tameem}[1]{\customComment{T}{#1}}

\title{New and Improved Concrete Lower Bounds for Orthogonal Vectors} 

 \author{Tameem Choudhury}{IIT Hyderabad, Hyderabad, India}{cs20resch11004@iith.ac.in}{https://orcid.org/0000-0002-5044-9717}{}
 \author{Nutan Limaye}{IT University of Copenhagen, Copenhagen, Denmark}{nuli@itu.dk}{https://orcid.org/0000-0002-0238-1674}{}
 \author{Karteek Sreenivasaiah}{University of Liverpool, Liverpool, UK}{karteek.sreenivasaiah@liverpool.ac.uk}{https://orcid.org/0000-0001-7396-3383}{}
 \author{Srikanth Srinivasan}{University of Copenhagen, Copenhagen, Denmark}{srsr@di.ku.dk}{https://orcid.org/0000-0001-6491-124X}{}

\authorrunning{T Choudhury, N Limaye, K Sreenivasaiah, S Srinivasan } 

\Copyright{CC-BY} 



\category{} 

\relatedversion{} 

\nolinenumbers 

\EventEditors{John Q. Open and Joan R. Access}
\EventNoEds{2}
\EventLongTitle{42nd Conference on Very Important Topics (CVIT 2016)}
\EventShortTitle{CVIT 2016}
\EventAcronym{CVIT}
\EventYear{2016}
\EventDate{December 24--27, 2016}
\EventLocation{Little Whinging, United Kingdom}
\EventLogo{}
\SeriesVolume{42}
\ArticleNo{23}

\begin{document}
	
	\maketitle
	
	\begin{abstract}
          The Orthogonal Vectors Problem ($\OVnd$) takes as input
          two sets $A,B$ each containing $n$ $d$-dimensional Boolean
          vectors, and outputs $1$ if and only if there exists
          $a \in A$ and $b \in B$ such that $a$ and $b$ are
          orthogonal. The $\OV$ conjecture states that for every
          $\varepsilon > 0$, there exists a constant $c \geq 1$ such
          that there is no algorithm deciding $\OVnd$ for
          $d = c \log n$ with running time $O(n^{2-\varepsilon})$. The analogous $\kayov$ conjecture hypothesizes a lower bound of $n^{k-\epsilon}$ for the same problem with $k$ sets. We prove these results and variants unconditionally in concrete computational models.

          \begin{itemize}
              \item We study a natural \emph{monotone}
          version of the $\kayov$-conjecture and shows that it holds for \emph{monotone} circuits and \emph{constant-depth} (not necessarily monotone) circuits when $d = n^{\Omega(1)}.$
          \item We show that the monotone version of the $\OV$ conjecture
          holds for \emph{monotone circuits}. More formally, we show
          that for every $\epsilon > 0$, there exists $c$ such that
          any monotone circuit family computing the negation of
          $\OVnd$ with $d=c\log n$ must have size
          $\Omega(n^{2-\epsilon})$.
          \item We also prove stronger Boolean formula and branching program lower bounds for $\OVnd$, strengthening a previous result of Kane and Williams (ITCS 2019). In particular, our Boolean formula lower bound of $\Omega(n^2 d)$ is tight up to constant factors.
          \end{itemize}


	\end{abstract}
	\section{Introduction}

        \noindent
        \textbf{Motivation.}
        The Orthogonal Vectors ($\OVnd$) problem is the
        computational problem of deciding if among two given sets $A$
        and $B$ of $n$ $d$-dimensional Boolean vectors, there exists a
        vector $a\in A$ and a vector $b\in B$ such that $a$ and $b$
        are orthogonal. This problem has received a lot of attention
        over the past decade for being able to capture the hardness of several
        important computational problems that are seemingly
        unrelated. Some important examples of these include \emph{edit
          distance} \cite{backursi15}, \emph{subset sum}
        \cite{abboudbhs22},  and \emph{longest common subsequence}
        \cite{abboudbw15}.  It has been established that if any of
        these problems has an algorithm that runs in truly
        sub-quadratic time, then so does $\OV$. For a more thorough treatment of this line of work, we refer the reader to the survey by Vassilevska Williams \cite{williams2018some}. 
        This body of research has
        led to the following natural conjecture about the time
        complexity of deciding $\OVnd$. 
         \begin{conjecture}[$\OV$ conjecture]
         \label{conj:OVC}
             For every $\epsilon > 0$,
        there exists $c\ge 1$ such that there is no algorithm
        (deterministic or randomized) deciding $\OVnd$ for
        $d = c\log n$ with running time
        $O(n^{2-\epsilon})$.
        \end{conjecture} 
        In particular, when
        $d\in \omega(\log n)$, the conjecture states that there is no truly subquadratic
        algorithm deciding $\OVnd$. It should be noted that when $d < \log n$, there are indeed
        sub-quadratic algorithms known \cite{Williams24}. Also, when $d = c\log n$, 
        Abboud, Williams and Yu \cite{abboudwyh15} show that $\OVnd$ 
        can be decided by a randomized algorithm in time $O(n^{2-1/O(\log c)})$. 
        This was later derandomized by Chan and Williams \cite{ChanWilliams21}.

        
        As positive evidence for the $\OV$ conjecture, Williams
        \cite{williams2005new} showed that the Strong Exponential Time
        Hypothesis (SETH) \cite{calabroip10,impagliazzop01} implies
        the $\OV$ conjecture. Thus, proving the latter can also be seen as a stepping stone to SETH.

        The fact that $\OV$ captures the hardness of a wide variety of problems makes it very
        interesting to ask if we can prove these conjectures
        unconditionally in restricted models of
        computation. The relevance of this line of research is twofold: the first, to rule out proposed algorithmic paradigms for these problems and the second, to find structural characteristics of the problem that will hopefully lead to an unrestricted lower bound in the future.\srikanth{Added a sentence to motivate. Maybe we should also look at Kane-Williams?}
        
        A recent line of
        research has been studying precisely this question. Kane and
        Williams \cite{kane2019orthogonal} showed strong lower bounds
        for the size of Boolean Formulas and Branching Programs computing
        $\OV$ thus proving the $\OV$-conjecture for these two
        models. Choudhury and Sreenivasaiah \cite{ChoudhuryTOCT25} show
        that this conjecture is true when
        computation is restricted to depth-3 $\ORANDOR$ circuits with
        constant bottom fan-in. 
        
        This latter work even proves such a result for the more general \emph{$\kayov$-conjecture}. In the $\kayov_{n,d}$ problem, the input is $k$ lists $A_1,\ldots, A_k \subseteq \{0,1\}^d$ of size $n$ each and the question is to decide if there exist vectors $a_1\in A_1,\ldots, a_k\in A_k$ such that $a_1\cap a_2\cap\ldots a_k = \emptyset$ (here we identify each Boolean vector $a_i\in \{0,1\}^d$ with a subset of $[d]$ in the natural way). The $\kayov$ conjecture, analogous to Conjecture~\ref{conj:OVC} is as follows. 
        
        \begin{conjecture}[$\kayov$ conjecture]
         \label{conj:kOVC}
             For every $\epsilon > 0$ and $k\geq 2$,
        there exists $c\ge 1$ such that there is no algorithm
        (deterministic or randomized) deciding $\kayov_{n,d}$ for
        $d = c\log n$ with running time
        $O(n^{k-\epsilon})$.
        \end{conjecture} 
        The complexity of $\kayov$ also has intimate connections with
        several important functions including graph diameter
        \cite{backursRSWW18} and longest common subsequence
        \cite{abboudbw15}. (See 
        \cite{williams2018some} for more.)\\

        \noindent
        \textbf{Results.} We now describe the results of this paper, which yield new and stronger statements along these lines.\\

        \noindent
        \textbf{Monotone and constant-depth Boolean circuits.} A major focus of this paper is on \emph{monotone}
        computation. Observe that the negation of $\OV$  (and more generally $\kayov$) is a monotone
        function.\footnote{A monotone function $f:\{0,1\}^n\rightarrow \{0,1\}$ is one that satisfies the property  $x\leq y \Rightarrow f(x)\leq f(y)$.} 
        We define the function $\kintnd:(\{0,1\}^d)^n\times (\{0,1\}^d)^n\to \{0,1\}$ as the negation of $\kayov_{n,d}$. 
        \begin{definition}[$\kintnd$] \label{def:Int}
            Fix any $k\geq 2.$
    		For tuples $A_1,A_2,\ldots,A_k \subseteq \{ 0,1\}^d$ with $|A_1| = \cdots = |A_k| = n$,
    		\begin{align*}
    			\kintnd(A_1,\ldots,A_k) = 1 \iff& \forall a_1\in A_1,\ \forall a_k\in A_k,\   \text{ we have }
    			 a_1 \cap a_2\cap \cdots a_k  \neq \emptyset. 
    		\end{align*}
            (where we interpret each Boolean vector $a_i$ as the characteristic vector of a subset of $[d]$.) In the particular case that $k= 2$, we use $\Intnd$ to denote $\kintnd.$
    	\end{definition}    
        In light of Conjectures~\ref{conj:OVC} and \ref{conj:kOVC}, it is natural to study the complexity of $\Intnd$ and $\kintnd$ for \emph{monotone} circuits\footnote{That is, Boolean circuits made up of only $\mathsf{AND}$ and $\mathsf{OR}$ gates. Such circuits can compute all monotone functions.} and prove that any monotone Boolean circuit for this problem must have size at least $n^{2-\epsilon}$ for $d = c\log n$ (where $c = c(\epsilon)$ is large enough). If Conjecture~\ref{conj:OVC} is true for general Boolean circuits, then this monotone lower bound must also hold. On the other hand, we note that the
        power of monotone circuits is likely to be incomparable with
        non-monotone models such as depth-$3$ circuits, Boolean formulas and branching programs studied in earlier work. In particular,
        there are $\cP$-complete problems that have small monotone
        circuits and it is unlikely that these can be simulated by
        formulas or branching programs in polynomial size (this can be
        proved unconditionally for depth-$3$ circuits). This
        makes the monotone question interesting.

        Another computational model we study in this paper is the class of \emph{constant-depth circuits}.\footnote{These are Boolean circuits made up of $\mathsf{NOT}$ gates and $\AND, \OR$ gates of unbounded fan-in. The depth refers to the amount of nesting in the circuit, or alternatively the length of the longest input to output gate path. } This circuit class can  implement many non-trivial algorithms such as the colour-coding algorithms for subgraph isomorphism~\cite{AYZ} and is incomparable with models studied in previous work (and provably stronger than the depth-$3$ model studied in~\cite{choudhuryS24}). 

        Our first result is a proof of a weak form of Conjecture~\ref{conj:kOVC} (and thus also a weak form of Conjecture~\ref{conj:OVC}) in the above two settings. The reason this is a `weak form' is that this proof works when the dimension $d$ of the input vectors is $n^{\Omega(1)}$.

        \begin{restatable}{theorem}{reductionthm}
        \label{thm:reductionthm}
          Fix any constants $k, h\geq 2$ and any constant $\epsilon > 0$. Any monotone circuit or any depth-$h$ (not necessarily monotone) circuit computing $\kintnd$ for $d = n^\epsilon$ must have size at least $n^{k-\epsilon}.$\footnote{Since the constant-depth circuit model is not monotone, we could also have stated the lower bound for this model directly for $\kovnd$. }
        \end{restatable}

        The constant-depth circuit lower bound is a resolution (in a weaker form) of a question due to Paturi stated in the work of Kane and Williams~\cite{kane2019orthogonal}.

         While the above theorem only proves lower bounds for relatively large values of $d$, it is worth noting that the $\OV$ problem is already interesting for fine-grained complexity in this setting~\cite{GIKW,ABDN}.

        \nutan{TODO (for nutan): Make the link between the problems explicit. (See comment by Reviewer 1)}

        Our next result is a complete resolution of the monotone Boolean circuit version of Conjecture~\ref{conj:OVC}.
        
        \begin{restatable}{theorem}{maintheorem}
        \label{thm:maintheorem}
          For every $0 <\epsilon <1$, there exists a constant
                $c$ such that for all $d \geq c\log n$ any monotone
                circuit computing $\Intnd$ requires at least
                $\Omega(n^{2-\epsilon})$ many gates.
        \end{restatable}
        \nutan{Can we add a small comment about why this completely resolved Conj 1. The connection between Int and OV and why anti-monotone lower bound also holds. This had confused the first set of reviewers.}
        \noindent
        \textbf{Boolean Formulas and Branching Programs.} In our third result, we revisit the Boolean formula and Branching Program lower bounds of Kane and Williams~\cite{kane2019orthogonal}. We improve the parameters in both results. In particular, this leads to a \emph{tight} lower bound (up to constant factors) for Boolean formulas over the full binary basis.\srikanth{We should instead write this for formulas of any constant fan-in. This is how Kane and Williams do it. Neciporuk still works.} More precisely, we show the following.

        \begin{theorem}[Boolean formula and Branching Program lower bound for $\OVnd$ (Informal)]
            \label{thm:Neciporuk}
            Assume $d > c\log n$ for $c > 1.$ Any binary Boolean formula computing $\OV_{n,d}$ has size at least $\Omega(n^2 d)$. Any Boolean branching program for $\OVnd$ has size at least $\Omega(n^2 d/\log (nd)).$
        \end{theorem}

        We note that this bound is tight \emph{up to constant factors}, as $\OVnd$ has a simple Boolean formula of size $O(n^2 d).$ Kane and Williams~\cite{kane2019orthogonal} prove a  weaker lower bound of $\Omega(n^2/\log d)$ in this regime of parameters. For Branching programs Kane and Williams~\cite{kane2019orthogonal} show a lower bound of $\Omega(n^2/(\log d \log nd))$. Our result improves this bound by a factor of $d \log d$.\\

        \noindent
        \textbf{Circuit constructions.}  Kane and Williams~\cite{kane2019orthogonal} observed that the standard brute-force algorithms for $\kayov$ can be carried out by constant-depth circuits.
        We note that for $\kintnd$, these circuits can also be made monotone. 
        We give two different brute-force constructions, one for small $d$, and another when $d$ 
        is large. In particular, for $d\in o(\log n)$, we show that a monotone
        formula of size $O(2^dn d)$ suffices to compute $\Intnd$. We note that 
        for such small $d$, Williams \cite{Williams24} shows a $\tilde{O}(1.35^dn)$ 
        time randomized algorithm for $\OV$. It is not clear if such a bound 
        can be achieved using monotone circuits.\\

        \subsection{Techniques}
        \label{sec:techniques}

        \noindent
        \textbf{Lower bounds via reductions.} The proof of Theorem~\ref{thm:reductionthm} is based on a simple split-and-list monotone reduction (inspired from~\cite{ABDN}) from graph and hypergraph versions of the clique\footnote{The Clique problem asks if a given input graph or hypergraph (presented as an adjacency matrix or tensor) contains a clique of a given size.} problem to $\kintnd$. The basic observation is that a set $S$ of vertices forms a clique of size $\ell$ in a graph (or hypergraph) if and only if it can be partitioned into $k$ subsets $S_1,\ldots, S_{k}$ of size $\ell/k$ each such that each of them is a clique and each edge contained in the union $S_1\cup \cdots \cup S_k$ is present in the graph (or hypergraph). This can be turned into a reduction from clique on $m$-vertex (hyper-)graphs by creating $k$ lists of size $\binom{m}{\ell/k}$ where each element of the $i$th list corresponds to a set $S_i$ of $\ell/k$ vertices and the corresponding vectors (which have dimension $m^{O(1)}$ each) are orthogonal if and only if $S_1\cup \cdots \cup S_k$ form a clique. 

        To obtain the required lower bound for $\kintnd$ from here, we need a strong lower bound ($m^{\ell\cdot (1-o(1))}$) for the clique problem on graphs or hypergraphs.
        In the monotone setting, such bounds have been known for a long time due to work of  Razborov~\cite{razborov85} (we use a statement from a paper of Alon and Boppana~\cite{AlonBoppana87}). Similar results are known for hypergraph versions of the clique problem via a result of Amano~\cite{amano10}. Along with the reduction, these results imply the desired lower bound. The reduction and its consequences appear in Section~\ref{sec:reduction2}.

        Abboud, Bringmann, Dell and Nederlof~\cite{ABDN} showed how to use the same high-level idea to reduce the problem of deciding $r$-partite clique on $r$-partite $r$-uniform hypergraphs to $\kintnd$.\footnote{They state their results in terms of $\kayov_{n,d}$ but we phrase it in terms of $\kintnd$ as we prefer to stay with monotone problems.} Unfortunately, the requisite circuit lower bounds are not known in the $r$-partite setting and so we need to modify their reduction.  The same work also gives reductions from the general $r$-uniform hypergraph clique problem to the $r$-partite case but it is unclear if such reductions can be implemented by monotone circuits.\\

        \noindent
        \textbf{OV conjecture for monotone circuits.} The above reduction gives us the required lower bound for monotone circuits computing $\Intnd$ when
        the dimension $d$ is 
        larger than the range prescribed in Conjecture
        \ref{conj:OVC}. The reasons for this have to do with
        existing monotone lower bounds for Clique not being tight
        in a fine-grained sense. Unfortunately, even the most optimistic monotone circuit lower bound for Clique (which is not known) would not yield a near quadratic lower bound for all $d = \omega(\log n)$. (This is described in more detail in
        Remark \ref{rem:disadvantages}.) Thus, more ideas are required to prove the $\OV$ conjecture in its strong form for monotone circuits.
        
        To do this, we turn to known techniques for proving strong monotone circuit lower bounds, which have been known since the
        1980s.  Razborov \cite{razborov85} developed a
        method to prove superpolynomial lower bounds on the
        size of monotone circuits computing Clique. This technique is now called the `approximation
        method'. Since Razborov's work, the approximation method has been
        the main approach to proving lower bounds against monotone
        circuits and has been refined and strengthened in several
        works such as \cite{AlonBoppana87, AmanoMaruoka05}. A
        crucial ingredient in the original approximation method is the
        Sunflower Lemma by Erd{\"o}s and Rado
        \cite{erdos1960}. Recently, Rossman \cite{rossman14} defined a
        more relaxed version of Sunflowers, now commonly known as
        ``robust Sunflowers'' and showed average case lower
        bounds for computing Clique. This has also been used to
        strengthen worst case monotone size lower bounds
    \cite{CavalarKumarRossman22}, \cite{cavalarGRSS25}\karteek{is
          this recent breakthrough paper published yet? I just checked, there is still no published version (5th May 2026)}.\srikanth{It will appear in STOC 2026.}


        To achieve better results than the one we get from the above
        reduction (Section~\ref{sec:monotoneOVC}), we apply a modified version of 
        Razborov's approximation method directly to the $\Intnd$ problem. For readers not familiar with
        the standard approximation method, we describe it briefly
        here. The idea is to start with a monotone circuit $C$ of size $s$ computing the target function $f$, and show that $C$ can be ``approximated'' closely using a small depth-two circuit (i.e. a DNF or a CNF)  $\tilde{F}$. Here
        by `approximation', we mean the behaviour of $\tilde{F}$ with
        respect to $C$ on two distributions $\Dz$ and $\Do$ on
        $0$-inputs and $1$-inputs of $f$ respectively. Then we show
        that on the one hand, any small depth-two circuit must have large error, say
        $\delta$, with respect to $f$. On the other hand, the DNF
        $\tilde{F}$ that we constructed has error at most $s\epsilon$
        with respect to $C$. Then we can conclude
        $s\ge \delta/\epsilon$.

        In the typical setting where the approximation method is used,
        the notion of `small' for the depth-two circuits has to do with the number of
        terms/clauses, and their maximum width. Ensuring that the
        approximator  $\tilde{F}$ is small usually involves a
        closure operation that uses the Sunflower Lemma
        or something similar \cite{razborov85, AlonBoppana87,rossman14}
        .
        In our  setting, all these are
        too demanding to be used. At a high level, this is because the marginal probabilities of our chosen distribution $\Dz$ are too low for such techniques to be applicable. See Remark~\ref{rem:marginal}.

        The main difference between the standard approximation method
        and ours is in the structure of the approximator. Instead of a
        DNF/CNF, our approximator is a conjunction of \emph{arbitrary}
        monotone Boolean functions that are each ``local'' to any one vector
        among the $2n$ vectors. i.e., each function takes inputs from
        at most one vector. We define this in Definition
        \ref{defn:approximator}.  This allows us to avoid expensive
        combinatorial closures such as the Sunflower lemma that are
        inapplicable. 
        A crucial property that we
        maintain for our approximator is that each function (in the
        conjunction) errs on $\Do$ with a non-trivial
        probability. Maintaining this property requires a technical
        lemma (Lemma \ref{lem:mupmuq}) that is similar in spirit to
        the Kruskal-Katona Theorem (\cite{kruskal63,katona68}) but in
        a much simpler setting where we work with a product
        distribution.\\

        \noindent
        \textbf{Formula and Branching Program lower bounds.} The source of improvement in the formula and branching program lower bounds is the same, so we describe the formula case. Consider the $\OVnd$ problem with input lists $A,B$. The basic idea of Kane and Williams, which fits into a classical framework for such lower bounds due to Nechiporuk~\cite{nechiporuk1966}, is to show that for any formula $F$ and each input vector $a\in A$, the bits of $a$ label many leaves of $F.$ This can be done by setting all the bits of the other elements of $A$ to $1$ and\footnote{This step is actually not required in the Kane-Williams argument as they consider the version of $\OV$ where $A = B$. The arguments are similar in both cases.} then setting the other elements of $B$ from  a fixed subset $X\subseteq \{0,1\}^d$ of size $n.$ For each subset $X'\subseteq X$, we get a restricted function on the bits of $a$ by using exactly the elements of $X'$ as assignments for the elements of $B.$ This gives $2^{n}$ different `sub-functions' on the bits of $a$ leading (by a simple counting argument) to a lower bound of roughly $n$ on the number of leaves labelled by the bits of $a.$

        To improve this argument, our main observation is that we do not need to fix the subset $X$ beforehand, and doing so unnecessarily reduces the number of sub-functions. Instead, by counting the number of subsets $X$, we obtain a bound of (roughly) $2^{nd}$ instead. This leads to a stronger (and in fact optimal) formula lower bound.

        These results are proved in Section~\ref{sec:Nechiporuk}.
        \subsection{Other related work}

        \noindent
        \textbf{$\kayov_{n,d}$ when $d$ is large.} As mentioned above, the $\kayov_{n,d}$ problem is already interesting in the setting when $d = n^{\Omega(1)}.$ This problem was introduced in work of Gao, Impagliazzo, Kolokolova and Williams~\cite{GIKW} who studied its connections to the fine-grained complexity of a large class of problems defined by first-order logic formulas. This was also further studied in the aforementioned work of Abboud et al.~\cite{ABDN} who made connections to weighted versions of hypergraph clique problems and optimization variants of SAT.\\

        \noindent
        \textbf{Jukna's criterion.}
        Another well-known route for proving monotone circuit lower bounds is via the monotone switching lemma formulated by Jukna~\cite{Jukna-monotone} and Berg and Ulfberg~\cite{BergUlfberg}. Unfortunately, the criterion used for lower bounds in these works does not seem to be applicable to our problem. See Section~\ref{sec:Jukna} for a justification.

	\section{Preliminaries}
	
    For formal definitions of Boolean circuits, formulas,
    and branching programs, we refer the reader to a standard text 
    such as Jukna \cite{juknabook}.

        For any $x,y \in \{0,1\}^d$, we write $x\leq y$ if
	$ \forall i, x_i \leq y_i$.
	\begin{definition}[Monotone function]
          We say that a Boolean function $f$ is monotone if
          $\forall x,y \in \{0,1\}^d$ such that $x \le y$, we have
          $f(x) \le f(y)$.
	\end{definition}

    We often interpret a $d$-dimensional vector $u\in \{0,1\}^d$ as the
	characteristic vector of a subset of $[d]$. For a vector $u_i$, we denote
    the $j$'th bit with $u_{ij}$.

    Recall that the dual of a Boolean function $f:\{0,1\}^n \rightarrow \{0,1\}$ is a Boolean function $f^*:\{0,1\}^n\rightarrow \{0,1\}$ defined by
    \[
    f^*(x_1,\ldots, x_n) := \neg f(\neg x_1,\ldots, \neg x_n).
    \]

    The following is standard.
    \begin{observation}
        \label{obs:boolean-dual}
        The dual $f^*$ of a monotone Boolean function $f$ is also monotone. Furthermore, for any $s$, $f$ has a monotone circuit of size $s$ if and only if $f^*$ has a monotone circuit of size $s.$
    \end{observation}

    In particular, we will sometimes consider the dual of $\Intnd$. The following is a consequence of the definition of $\Intnd$.

    \begin{observation}
        \label{obs:dual-kint}
        For tuples $A,B\subseteq \{0,1\}^d$ with $|A| = |B| =  n$,
        \begin{align*}
            \Intnd^*(A,B) = 1 \Longleftrightarrow &\ \exists a\in A, \exists b \in B \text{ such that }
             a\cup b = [d].
        \end{align*}
    \end{observation}

    \begin{definition}[$\ell$-clique]
      The problem $\ell$-clique is
      that of deciding if a given graph
      contains a clique (complete graph) on a subset of at least $\ell$
      vertices.

      When the parameter $\ell$ is not important, we will
      use $\clique$ to refer to the $\ell$-clique problem.
    \end{definition}
    
    We will use $\mu_p$ to denote the distribution where each bit is set 
    independently at random to $1$ with probability $p$.

\nutan{Todo: Add definitions of formula and BP and also L() and BP().}\karteek{Notation L() and BP() have been defined in the Nechiporuk section, 2nd paragraph.}\karteek{Done}

        \section{Lower bounds via reductions from Hyperclique}
    \label{sec:reduction2}
    \srikanth{New section with the stronger reduction. Can replace Section~\ref{sec:reduction} if correct.}

    To state the main result of this section, we start with a definition. 

    For parameters $m,r,\ell$, let $\hypcliquemrl: \{0,1\}^{{{m}\choose{r}}} \rightarrow \{0,1\}$ be a Boolean function that takes as input the characteristic vector of the edge set of an $r$-regular hypergraph on $m$ vertices and outputs a $1$ if and only if it has a hyperclique of size $\ell$.

    \begin{theorem}
    \label{thm:redn}
    The following holds for any small enough constant
    $\varepsilon > 0$ and positive integer constants $r,k, h$. Let $m$ be a parameter and let
    $\ell := \ell(m)$ be a non-decreasing function such that
    $\ell = o(m)$, $\ell$ is divisible by $k$. Suppose any monotone circuit (resp. depth-$h$ (not necessarily monotone) circuit) computing $\hypcliquemrl$
    has size at least
    $\binom{m}{\ell}^{1-\varepsilon}$. Then for
    $n = \binom{m}{\ell/k}$, there exists $d = O(m^r)$ such that any monotone circuit (resp. depth-$h$ circuit) for
    $\kintnd$ on length-$n$ lists of vectors of dimension $d$ has size
    at least $n^{k\cdot (1-2\varepsilon)}.$
    \end{theorem}

    Before we prove the theorem above, we use it to derive the monotone lower
    bound and constant-depth lower bound for computing $\kintnd.$

    \paragraph*{Monotone circuit lower bound.}
    For the monotone lower bound, we will use the following lower
    bound stated by Alon and Boppana~\cite{AlonBoppana87} following
    the work of Razborov~\cite{razborov85}.

    \begin{theorem}[\cite{AlonBoppana87, razborov85}]
        \label{thm:AB87}
        Any monotone circuit solving $\ell$-Clique on graphs with
        $m$ vertices has size at least
        $\frac{m^\ell}{2^{O(\ell^2)}\cdot (\log m)^\ell}.$ In
        particular, for any large enough $\ell = \ell(m) = o(\log m),$ the
        lower bound is $m^{\ell\cdot (1-o(1))}.$
    \end{theorem}

    \begin{corollary}\label{cor:reduction}
      Fix any constant $\epsilon > 0$. Let $n$ be a growing parameter
      and $\epsilon > 0$ any constant. There is a function
      $d(n) = \exp((\log n)^{1/2 + o(1)}) = n^{o(1)}$ such that for infinitely
      many $n,$ any monotone circuit computing $\kintnd$ must have
      size at least $n^{k\cdot (1-\epsilon)}.$
    \end{corollary}
    \nutan{Do we mean $\kintnd$ here?}\karteek{Yes. Fixed!}
    \begin{proof}[Proof of Corollary~\ref{cor:reduction}]
      The proof is an immediate consequence of Theorem \ref{thm:redn}
      along with Theorem~\ref{thm:AB87} above. More precisely, we fix
      $\ell$ to be a growing parameter of $m$ such that
      $\ell = \log m/\log\log m$ and $\ell$ is divisible by $k$. We will apply
      Theorem~\ref{thm:redn} with $\delta = \epsilon/2$. \nutan{The difference between vareps and eps is too small. Maybe use a visibly different parameter?}
        
      Note that Theorem~\ref{thm:AB87} yields a monotone circuit lower
      bound for clique of size
      $m^{\ell\cdot (1-o(1))} \geq \binom{m}{\ell}^{1-\delta}.$
      We set $n = \binom{m}{\ell/k} = \exp((\log m)^{2-o(1)})$ and 
      $d\leq m^2 = \exp(O(\log m)) = \exp((\log n)^{1/2 + o(1)}).$ \nutan{I don't see why $\binom{m}{\ell/k} = \exp((\log m)^{k-o(1)})$. Can you justify?}\karteek{Corrected exponent to 2}
      For
      this setting of parameters, Theorem~\ref{thm:redn} implies a
      lower bound of $n^{k(1-2\delta)} = n^{k(1-\epsilon)}.$
      \nutan{The last equality will seem very strange if a reader misses the diff between vareps and eps. :)}
    \end{proof}

    
    \paragraph*{Constant-depth lower bound.} For the constant-depth lower bound, we use the following theorem by Amano~\cite{amano10}. 
    We will need the following definition to state the theorem.

    \begin{theorem}
        \label{thm:amano}
        For every constant $\ell > r > 2$ every depth-$h$ circuit computing $\hypcliquemrl$ must have size at least  
        $\Omega(m^{\ell \cdot (1-\frac{\log r + 2}{r-1})}).$ In particular, for constant $r$, the lower bound is $m^{\ell \cdot (1-\delta)}$ for a constant $\delta$.

    \end{theorem}

    Now, using Theorem~\ref{thm:redn} and Theorem~\ref{thm:amano}, we get the folowing corollary for depth-$h$ circuits which proves  Theorem~\ref{thm:maintheorem}. 
    \begin{corollary}\label{cor:reductiion-depth-h}
      Fix any constant $\epsilon > 0$. Let $n$ be a growing parameter
      and $\epsilon > 0$ any constant. There is a function
      $d(n) = \exp((\log n)^{1/2 + o(1)}) = n^{o(1)}$ such that for infinitely
      many $n,$ any depth-$h$ circuit computing $\kintnd$ must have
      size at least $n^{k\cdot (1-\epsilon)}.$
    \end{corollary}
    \begin{proof}[Proof of Corollary~\ref{cor:reductiion-depth-h}]
        For any fixed $k$, and for $r$ at least $\frac{1}{\delta} \cdot \log (1/\delta)$, we choose $\ell$ such that $\ell$ is divisible by $k$ and $\ell$ is at least $kr/2\delta$. Note that, for this setting and for $n = \binom{m}{\ell/k}$ we have $d = n^{o(1)}$. Now, we can invoke Theorem~\ref{thm:redn}. Note that, we have a lower bound $m^{\ell(1-\delta)}$ from~\cite{amano10}, which is stronger than $\binom{m}{\ell}^{(1-\delta)}$, which we need to apply Theorem~\ref{thm:redn}. We get a lower bound of $n^{k (1-2\delta)}$. This completes the proof. 
    \end{proof}

  \begin{proof}[Proof of Theorem~\ref{thm:redn}] 

    The proof is the same for both monotone circuits and depth-$h$ circuits. 

    Let $G=(V,E)$ be an $r$-uniform hypergraph on $m$ vertices. Let
    $x$ be the bit-vector that represents this hypergraph when given as input to a circuit. That is, the
    entries of $x$ are labelled by sets $e\in \binom{V}{r}$ and
    $x_e = 1$ if and only if the edge $e$ is present in the hypergraph
    $G$. Informally, $x$ is the \emph{adjacency tensor} of $G.$
    
        We first describe how to construct, for any
        $S \subseteq V$ and $i\in [r]$, vectors $u_{i,S}$ of length $d'$ (defined below) such that
        for all \emph{pairwise disjoint} $S_1,\ldots, S_k\subseteq V$, we have:
        \begin{align*}
          u_{1,S_1}\cup\cdots\cup u_{k,S_k} = [d'] \iff S_1\cup\cdots\cup S_k \text{ forms an $r$-hyperclique in }G
        \end{align*}
        Define\karteek{To be proper pedantic, $\mcI$ should really be $\mcI_r$. But
          we never use it in a different context, so i am leaving this
          without subscript.}
          $\mcI = \{(c_1,\ldots, c_k)\mid \forall i~ c_i\in
          \mathbb{N}\cup\{0\}\text{ and } \sum_ic_i = r \}$.  We
          define $u_{i,S}$ with dimension
          $d' = \binom{m}{r}\cdot |\mcI|$ as follows.  For
          $e = \{v_1,\ldots, v_r\}\subseteq V$ and
          $(c_1,\ldots, c_k)\in \mcI$,
        \begin{align*}
          u_{i,S}[e,(c_1,\ldots, c_k)] = \begin{cases}
            0 & \text{if } c_i = 0\\
            1 & \text{if } |e\cap S| \neq c_i\\
            x_e & \text{otherwise}             
            \end{cases} 
        \end{align*}
        \textbf{Correctness.} Let $S_1,\ldots, S_k\subseteq V$ be
        pairwise disjoint, and
        $e = \{v_1,\ldots, v_r\}\subseteq S_1\cup\cdots\cup S_k$ be
        chosen arbitrarily. Define
        $\vec{\alpha}\in (\mathbb{N}\cup\{0\})^k$ as
        $\alpha_i = |e\cap S_i|$. Since the sets $S_i$ are pairwise
        disjoint, $\sum_i \alpha_i = |e| = r$.

        To see the forward direction, suppose $e\not\in E$. i.e.,
        $x_e=0$. Then for each $i$ such that $\alpha_i = 0$, we have
        $u_{i,S_i}[e,\vec{\alpha}]=0$. For those $i$ where
        $\alpha_i\neq 0$, we know $\alpha_i = |e\cap S|$ by
        definition, and thus $u_{i,S_i}[e,\vec{\alpha}]= x_e =
        0$. Hence $(u_{1,S_1}\cup\cdots\cup u_{k,S_k})[i] = 0$.

        For the reverse direction, suppose $S_1\cup\cdots\cup S_k$
        indeed forms an $r$-hyperclique in $G$. Let
        $e= \{v_1,\ldots, v_r\}\subseteq S_1\cup\cdots\cup S_r$, and
        $\vec{c}:=(c_1,\ldots, c_r)\in \mcI$ be chosen arbitrarily. If
        $c_i = \alpha_i$ for all $i$, then for any choice of
        $j\in [k]$ such that $\alpha_j\neq 0$, we have
        $u_{j,S_j}[e,\vec{c}] = x_e = 1$.  Else
        $\vec{c}\neq \vec{\alpha}$. Since
        $\sum_i \alpha_i = \sum_i c_i = r$, there must exist
        $j\in [k]$ such that $c_j\neq \alpha_j$ (equivalently
        $c_j\neq |e\cap S_j|$) and $c_j\neq 0$.  Then
        by definition, we have $u_{j,S_j}[e,\vec{c}] = 1$. 

        For $e\not\subseteq S_1\cup \cdots S_r$ and any choice of $\vec{c}\in \mathcal{I}$, there must be an $i$ such that $c_i \neq 0$ and $|e\cap S_i| \neq c_i$ (since otherwise, we would have $|e\cap (S_1\cup \cdots \cup S_r)| = r$, a contradiction). This means that $u_{i,S_i}[e,\vec{c}]=1$ as in the previous paragraph.

       \paragraph*{The reduction.}
       We construct a monotone circuit $R$ that takes as input an
       $r$-uniform hypergraph $G = (V,E)$ with $|V|=m$ encoded as
       variables $\{x_{S} \mid S\subseteq V\text{ and } |S|=r\}$ and
       produces as output $k$ lists $A_1,\ldots, A_k$ of
       $n = \binom{m}{\ell/k}$ Boolean vectors each of dimension
       $d = \binom{k}{2}\cdot m + d' $ such that
        \[
        \kintnd^*(A_1,\ldots, A_k) = 1 \iff \text{$G$ has an $r$-hyperclique of size $\ell$.}
        \]
        The idea is to use the fact that a set $S\subseteq V$ forms an
        $r$-hyperclique of size $\ell$ in $G$ if and only if $S$ can
        be partitioned into $S_{1},\ldots, S_{k}$ each of size
        $\ell/k$ such that $S_1\cup\cdots\cup S_k$ is an
        $r$-hyperclique in $G$.

        The vectors in $A_i$ will be potential candidates for the set
        $S_i$. We will think of each such vector as having two
        parts. The first part will contain $\binom{k}{2}$ blocks (each
        of $m$ bits) that will be used to verify that
        $S_{1},\ldots, S_{k}$ are pairwise disjoint. The second part
        uses the $u_{i,S}$ vectors constructed earlier to check that
        $S_1\cup\cdots\cup S_{k}$ forms an $r$-hyperclique.

        We define the vectors in $A_1,\ldots, A_k$ as follows. The vector
        $a_{i,S}$ is the vector $u_{i,S}$ prefixed by 
        $\binom{k}{2}$ blocks each of $m$ bits. These are indexed by
        $\{j_1,j_2\}\in \binom{[k]}{2}$ and defined as:
        \begin{align*}
          a_{i,S}[\{j_1,j_j\}] =
          \begin{cases}
            \chi_{\overline{S}}& \text{ if } i\in \{j_1,j_2\}\\
            0 &\text{ otherwise.}                                 
          \end{cases}
        \end{align*}
        where $\chi_{\overline{S}}$ is the characteristic vector of
        $\overline{S}$.  As mentioned earlier, the above
        $\binom{k}{2}m$ bits are followed by the vector $u_{i,S}$.
        This completes the definition of $a_{i,S}$.
        

        \textbf{Correctness.} We show that $G$ contains an
        $r$-hyperclique of size $\ell$ if and only if
        $(A_1,\ldots, A_k)$ is a $1$-instance of $\kintnd^*$.

        \noindent
        ``$\implies$'': Suppose $G$ contains an $r$-hyperclique of size $\ell$
        on $T\subseteq V$.  Let $T_1,\ldots, T_k\subseteq T$ be an
        arbitrary partition of $T$ into $k$ parts each of size
        $\ell/k$.  We claim that
        $a_{1,T_1}\cup \cdots \cup a_{k,T_k} = [d]$. 

        To see that the bits in the first $\binom{k}{2}$ blocks are
        all $1$s, choose any index $i$ in this range. The index $i$ is
        in a block $\{j_1,j_2\}\in \binom{k}{2}$. Then exactly two of
        the vectors among $a_{1,T_1},\ldots, a_{k,T_k}$ are non-zero
        in the block $\{j_1,j_2\}$, namely $a_{j_1,T_{j_1}}$ and
        $a_{j_2,T_{j_2}}$. Recall that the $T_i$ are pairwise
        disjoint.  In particular, $T_{j_1}$ and $T_{j_2}$ are
        disjoint. Then by definition
        $a_{j_1,T_{j_1}}[\{j_1,j_2\}] = \chi_{\overline{T_{j_1}}}$ and
        $a_{j_2,T_{j_2}}[\{j_1,j_2\}]=
        \chi_{\overline{T_{j_2}}}$. Hence
        $a_{j_1,T_{j_1}}[\{j_1,j_2\}]\cup a_{j_2,T_{j_2}}[\{j_1,j_2\}]
        = \vec{1}$ which implies
        $(a_{1,T_1}\cup \cdots \cup a_{k,T_k})[i] = 1$. Since the
        choice of $i$ was abitrary, the first part of $a_{1,T_1}\cup \cdots
        \cup a_{k,T_k}$ is all $1$s.

        The second part of each $a_{i,S_i}$ is simply $u_{i,S_i}$.
        Since the $T_i$ are pairwise disjoint and
        $T = T_1\cup\cdots\cup T_k$ is an $r$-hyperclique, we have
        $u_{1,S_1}\cup\cdots\cup u_{k,S_k} = [d']$. Thus,
        $a_{1,T_1}\cup \cdots \cup a_{k,T_k} = [d]$.
        
        \noindent
        ``$\impliedby$'':\karteek{I think it should be possible to
          combine the two directions of this proof}
        Suppose $(A_1,\ldots, A_k)$ is
        a $1$-instance of $\kintnd^*$. Then there exists
        $a_{1,S_1},\ldots, a_{k,S_k}$ such that
        $a_{1,S_1}\cup\cdots\cup a_{k,S_k} = [d]$. We will show that
        $S_1\cup\cdots\cup S_k$ forms an $r$-hyperclique of size
        $\ell$ in $G$.

        Observe that in the first $m\cdot \binom{k}{2}$ bits, for any
        block $\{j_1,j_2\}\in \binom{[k]}{2}$, exactly two of the
        vectors among $a_{1,S_1},\ldots, a_{k,S_k}$ are non-zero
        namely $a_{j_1,S_{j_1}}$ and $a_{j_2,S_{j_2}}$. Since
        $a_{j_1,S_{j_1}}[\{j_1,j_2\}]\cup a_{j_2,S_{j_2}}[\{j_1,j_2\}]
        = \vec{1}$, we have 
        $\chi_{\overline{S_1}}\cup \chi_{\overline{S_2}} = \vec{1}$. This means
        $S_{j_1}$ and $S_{j_2}$ are disjoint. This is true for every choice of
        $\{j_1,j_2\}$, and hence all of $S_1,\ldots, S_k$ must be pairwise disjoint.

        Since $a_{1,S_1}\cup\cdots\cup a_{k,S_k} = [d]$, and the
        second part of each $a_{i,S_i}$ equals $u_{i,S_i}$ by
        definition, it must be the case that
        $u_{1,S_1}\cup\cdots\cup u_{k,S_k} = [d']$. Combining this with the fact
        that the $S_i$ are pairwise disjoint, we can conclude that
        $S_1\cup\cdots\cup S_k$ forms an $r$-hyperclique in $G$.

        \textbf{Analysis.}  Finally, we argue about the
        complexity of the reduction. Note that $u_{i,S}$ have dimension $d'$, where
        $d' = \binom{m}{r}|\mcI| = \binom{m}{r}\cdot
        \binom{r+k-1}{k-1}$ and thus
        $d = \binom{k}{2}\cdot m + \binom{m}{r}\cdot
        \binom{r+k-1}{k-1}$, which is the dimension of $a_{i,S}$ vectors.

        Furthermore, each coordinate of $u_{i,S}$ can be computed as a simple monotone projection (i.e. either a constant or a single variable) from the input graph. The numbers of such projections is $d'$. Thus, overall, the $u_{i,S}$ vectors can be computed by a circuit of size $O(m^r)$ (recall that $k,r$ are constants). The first part of $a_{i,S}$ is also a projection and thus can be computed with size $O(m)=$. Thus, the overall size of the projection is $O(\binom{m}{\ell/k} \cdot \left(m + m^r\right))$. 

        Thus, if $\kintnd$ has size $s$ monotone (resp. depth-$h$) circuit, then we get a monotone (resp. depth-$h$) circuit for $\hypcliquemrl$ of size $s + O(\binom{m}{\ell/k} \cdot \left(m + m^r\right)$. 

        By our assumed hypothesis on the circuit size of $\hypcliquemrl$, we get the following bounds.

        \begin{align*}
      s &\geq \binom{m}{\ell}^{1-\varepsilon} - O\left(\binom{m}{\ell/k} \cdot \left(m + m^r\right)\right)
          \geq \left(\frac{m}{\ell}\right)^{(1-\varepsilon)\ell} - O\left(\binom{m}{\ell/k}\cdot m^r\right)\\
        &= \left(\frac{me}{\ell/k}\right)^{k\cdot\frac{\ell}{k}\cdot (1-\varepsilon)}\cdot\frac{1}{(ke)^{\ell(1-\varepsilon)}}
          - O\left(\binom{m}{\ell/k}\cdot m^r\right)\\
        &\geq \binom{m}{\ell/k}^{k(1-\varepsilon)}\cdot\frac{1}{2^{O(\ell)}}
          - O\left(\binom{m}{\ell/k}\cdot m^r\right)\\
        &\geq \binom{m}{\ell/k}^{k(1-2\varepsilon)} = n^{k (1-2\varepsilon)}&
    \end{align*}  
    where the final inequality uses the fact that $\ell$ is large
    enough, $\ell = o(m)$, and constant $r$, which implies that for any small
    enough constant $\varepsilon$ and large enough $m$
    \[
      2^{O(\ell)}\leq \binom{m}{\ell/k}^{o(1)} \text{ and }
      \binom{m}{\ell/k}\cdot m^r \leq \binom{m}{\ell/k}^{k\cdot
        (1-2\varepsilon)}
    \]
    This concludes the proof. \end{proof}

\noindent
    \begin{remark}[Shortcomings of this proof framework]
    \srikanth{Need to revise this remark to include the $\mathrm{AC}^0$ lower bound if we add it.}
    \label{rem:disadvantages}
    While Theorem~\ref{thm:redn} yields a simple proof of the monotone
    hardness of $\Intnd$, we note that it only holds in the regime of
    moderately large $d$ (i.e. $d = \exp((\log n)^{\Omega(1)})$). This
    is due to the fact that we need a near-tight lower bound
    (i.e. $\binom{m}{\ell}^{1-\varepsilon}$) for $\ell$-clique. We
    only have such lower bounds in the regime of small
    $\ell = o(\log m).$ While we have strong lower bounds for larger
    values of $\ell$~\cite{AlonBoppana87, CavalarKumarRossman22,
      BlasiokCCC25}, they stop being near-tight in this sense.
    
    Though the $\Intnd$ problem is already interesting in this regime
    of parameters~\cite{ABDN}, the importance of $\Intnd$ for many
    fine-grained problems (e.g. SETH) lies in the setting of small $d.$ In this setting, even assuming the best possible lower bounds for
    clique, e.g. a near-tight lower bound in the entire regime where
    $\ell = o(m)$,\footnote{Current known lower
      bounds~\cite{BlasiokCCC25,deRezendeVinyalsCCC25} only work for $\ell \leq \sqrt{m}$, and are furthermore not tight enough for the above reduction to be applicable  when $\ell = \omega(\log m)$.}
    this would only yield a lower bound for $\Intnd$ in the regime
    where $d = \Omega((\log n)^2)$.
    \end{remark}

        \section{Monotone version of OV Conjecture}\karteek{Rename section to something else}
        \label{sec:monotoneOVC}
        In this section, we prove Theorem~\ref{thm:maintheorem}. \nutan{Todo: Add a rough outline.}
        
        \subsubsection*{Distributions for $0$-inputs and $1$-inputs}
        We define a distribution $\mathbf{D_0}$ over $0$-inputs of $\Intnd$ as follows:
        \begin{enumerate}
        \item\label{step:oneD0} Choose $i,j\in [n]$ uniformly and independently at random.
        \item Choose a vector $b\in \strings$ by choosing
          $b_1,\ldots, b_d \in \{0,1\}$ uniformly and independently at
          random.
        \item Set $u_i$ to $b$ and $v_j$ to $\overline{b}$.
        \item \label{step:allones} For each $k\neq i$, set vector $u_k$ to all $1$s. For
          each $k'\neq j$, set vector $v_{k'}$ to all $1$s.         
        \end{enumerate}
        Observe that setting $u_i = b$ and $v_j=\overline{b}$ makes them orthogonal, and thus
        every input drawn from $\mathbf{D_0}$ is a $0$-input of $\Int$.
        
        We define the distribution $\mathbf{D_1}$ as $\mu_p$, the distribution where each bit is
        set to $1$ independently at random with probability $p$. We observe that inputs drawn
        from $\mathbf{D_1}$ are $1$-inputs of $\Int$ with very high probability:
        \begin{lemma} \label{lem:pr1}
           \begin{align*} \label{lem:ge1}
           	\Pr_{(\mcU,\mcV)\sim \mathbf{D_1}}[\,\Intnd(\mcU,\mcV) = 1\,] \geq 1 -  n^2 e^{-p^2d}
           \end{align*} 
        \end{lemma}
        \begin{proof}
            Let $x,y$ be $d$-bit strings drawn from the distribution $\mu_p$. The probability that $x$ and $y$ are disjoint is $(1-p^2)^d$. By definition, $\Int(\mcU,\mcV) = 0$ if and only if there exists at least one pair $u\in\mcU$ and $v\in \mcV$ that are disjoint. Using a union bound over all $n^2$ possible pairs of vectors, we get $\Pr[\,\Int(\mcU,\mcV) = 0\,] \leq n^2 (1-p^2)^d \leq n^2e^{-p^2d}$.
        \end{proof}

        \noindent
        \textbf{Proof Outline.} As described in Section~\ref{sec:techniques}, the idea behind the approximation method which we use to prove Theorem~\ref{thm:maintheorem} is to show that any small monotone circuit $C$ defined on the inputs of $\OVnd$ can be approximated, w.r.t. the two distributions $\mathbf{D_0}$ and $\mathbf{D_1}$ defined above, by a monotone $F$ from a well-defined family of `simple' monotone functions that we will call \emph{approximators}. Once this is done, we can easily show that $C$ cannot have been computing $\OVnd$ since no approximator can even approximately compute $\OVnd$ w.r.t.  $\mathbf{D_0}$ and $\mathbf{D_1}.$ The hard part is to prove that the circuit $C$ indeed has an approximator with low error. This is done carefully at each gate of the circuit, replacing it with a suitable approximator in a way that does not increase the error too much. 

        The conceptually new part of the lower bound of this paper is the definition of the approximator as more `standard-issue' constructions do not seem to work (see Remark~\ref{rem:marginal} below). Showing that this new kind of approximator is amenable to the inductive argument requires some technical work, with the highlight being a simpler version of the Kruskal-Katona theorem (Lemma~\ref{lem:mupmuq} below) that allows us to bound the error at OR gates effectively.

        \begin{remark}
            \label{rem:marginal} In this technical remark, we sketch why the approximation method applied with more `standard' approximators do not seem to work in our setting here. We will assume knowledge of the  standard applications of the approximation method. 
            
            Note that the marginal probabilities in $\Dz$ are $1/2$ (and cannot be higher for a similar construction) conditioned on a vector $u_i$ being chosen in Step 1 of the sampling procedure. This technical point prevents us from using `standard' techniques such as CNF approximators (which would otherwise be very natural) as in e.g.~\cite[Section 3]{CFMSY} in what follows. To apply these ideas, we would need to be able to trade off the probability that a clause of width $w+1$ is set to $0$ under $\Dz$ (roughly $n^{-1}2^{-(w+1)}$) with the number of clauses required to make a sunflower under $\Do$, which is roughly $w^w$ using the best known robust sunflower lemmas~\cite{RaoSunflowerSurvey}. Unfortunately, this means that removing all clauses of width more than $w$ incurs an error bound of up to $w^w\cdot n^{-1}2^{-O(w)}$ which we need to be bound by $n^{-(2-\epsilon)}$ for a union bound over the gates of the circuit. Unfortunately, there is no choice of $w$ that satisfies this constraint. This renders the standard technique of using such approximators inadmissible in this context.
        \end{remark}

        We now begin the main proof.

        Recall that the inputs to $\Intnd$ are two $n$-tuples $\mcU$
        and $\mcV$ of $d$-dimensional vectors. In this context, 
        we call a function $B:\{0,1\}^d\to\{0,1\}$ a
        `\emph{local function}' if the inputs to $B$ are a vector $u\in \mcU$
        or a vector $v\in \mcV$. i.e., the function $B$ does not take inputs from more than one vector. 
        \nutan{I think the notation $\var$ is a bit confusing. $\var(B)$ refers to $u_i$s or $v_i$s, but these are Boolean vectors. We are also interpreting them as variables. It is a minor point though.}\karteek{Check now?} \nutan{This reads well!}
        \srikanth{Replace `bubble function' by `local function'.}\karteek{Done}

        We are now ready to describe the structure of our \emph{approximator}.
        \begin{definition}
          \label{defn:approximator}
          Let $0\le m\le n$ and $0\le \delta\le 1$ be any fixed
          parameters. A Boolean function
          $\tilde{F}:(\{0,1\}^d)^n\times (\{0,1\}^d)^n\to \{0,1\}$ is
          an $(m,\delta)$-approximator if it is a constant, or it can be
          written as
          \[ \tilde{F} = \bigwedge_{i\in S_\mcU} B_i \wedge
            \bigwedge_{j\in S_{\mcV}} C_j \] where
          $S_\mcU,S_\mcV\subseteq [n]$ and the following three properties are satisfied:

          \begin{enumerate}
          \item\label{pty:bubble} For each $i\in S_\mcU$ ($j\in S_\mcV$), the function $B_i$ ($C_j$)
            is a local function on $u_i$ ($v_j$ respectively).
          \item\label{pty:card} $|S_{\mcU}|\leq m$ and $|S_{\mcV}|\leq m$.
          \item\label{pty:error} For all $i\in S_\mcU$, and all $j\in S_\mcV$, we have:
            \begin{align*}
              \Pr_{x\sim \Do}[B_i(x)=0] \geq \delta  \text{ and }
              \Pr_{x\sim \Do}[C_j(x)=0] \geq \delta \\
            \end{align*}
          \end{enumerate}

        \end{definition}

        \subsubsection*{Constructing an $(m,\delta)$-approximator}
        \label{sec:construction}
        We begin by showing how to construct an
        $(m,\delta)$-approximator for a monotone circuit for any
        $\delta < (1-p)$ inductively starting at the leaves. We denote the approximator
        at gate $g$ with $\apr{g}$.

        \begin{itemize}
          
        \item \textbf{Leaves:} The leaves of the circuit are labelled by either variables or
          constants. For a variable $x_i$, the approximator is simply
          the function $x_i$. It is straightforward to see that
          properties \ref{pty:bubble} and \ref{pty:card} are satisfied
          in this case. To see property \ref{pty:error}, observe that
          the probability over $\Do$ that $x_i$ is $1$ is $p$. Since
          $\delta < (1-p)$, property \ref{pty:error} is satisfied. For a
          leaf that is labelled by a constant $0$ or $1$, the
          approximator is also the same constant.
          
        \item $\mathbf{g\wedge h}$:\\ Let
          $\apr{g}$ and
          $\apr{h}$ be
          $(m,\delta)$-approximators for $g$ and $h$ respectively. 
          It is easy to see that $\apr{g}\wedge\apr{h}$ satisfies properties \ref{pty:bubble}
          and \ref{pty:error} as 
          conjunctions of local functions are also  local functions, and for any two functions
          that satisfy property \ref{pty:error}, their conjunction
          will also satisfy property \ref{pty:error}. 
          If $\apr{g}\wedge\apr{h}$ also satisfies property \ref{pty:card}, then 
          we define $\apr{g\wedge h} := \apr{g}\wedge\apr{h}$. Otherwise define $\apr{g\wedge h} := 0$.\karteek{This is now simplified because of a reviewer comment. The formal notation from here is moved to the error bound analysis.}

        \item $\mathbf{g\vee h}$:\\
          Let
          $\apr{g} = \bigwedge_{i\in S_{\mcU}^g} B_i^{g} \wedge
          \bigwedge_{j\in S_{\mcV}^g} C_j^g$ and
          $\apr{h}=\bigwedge_{i\in S_\mcU^h} B_i^{h} \wedge
          \bigwedge_{j\in S_{\mcV}^h} C_j^{h}$ be
          $(m,\delta)$-approximators for $g$ and $h$ respectively. Then we have:
          \begin{align}
            \label{eq:apprvee}
            \apr{g}\vee \apr{h} =& \bigwedge_{i\in S_{\mcU}^g, i'\in S_{\mcU}^h} B_i^{g}\vee B_{i'}^{h} \wedge
            \bigwedge_{i\in S_{\mcU}^g, j'\in S_{\mcV}^h} B_i^{g}\vee C_{j'}^{h} \nonumber \\
            \wedge& \bigwedge_{j\in S_{\mcV}^g, i'\in S_{\mcU}^h} C_i^{g}\vee B_{i'}^{h}
                    \wedge \bigwedge_{j\in S_{\mcV}^g, j'\in S_{\mcV}^h} C_{j}^{g}\vee C_{j'}^{h}
          \end{align}

          From the above, we include in our approximator only those
          functions that are local functions and satisfy property
          \ref{pty:error}. This is done as follows. Let
          $T_\mcU=\{i\in S_{\mcU}^g\cap S_{\mcU}^h ~|~ (B_i^{g}\vee
          B_{i}^{h})\text{ satisfies property \ref{pty:error}} \}$ and
          $T_\mcV=\{j\in S_{\mcV}^g\cap S_{\mcV}^h ~|~ (C_j^{g}\vee
          C_{j}^{h})\text{ satisfies property \ref{pty:error}} \}$.
          If $|T_\mcU| = |T_\mcV| = 0$, then we define $\apr{g\vee h} = 1$. Else
          we define $\apr{g\vee h}$ as
          \[ \apr{g\vee h} = \bigwedge_{i\in T_\mcU} (B_i^{g}\vee B_{i}^{h})
            \wedge \bigwedge_{j\in T_\mcV} (C_{j}^{g}\vee
            C_{j}^{h}) \]

          Observe that since each function in this conjunction is a
          local function, $\apr{g\vee h}$ satisfies property
          \ref{pty:bubble}.  Also note that
          $|T_\mcU| \le |S_{\mcU}^g\cap S_{\mcU}^h|\le |S_\mcU^g|\le
          m$, and similarly
          $|T_\mcV| \le |S_{\mcV}^g\cap S_{\mcV}^h|\le |S_\mcV^g|\le
          m$. So, $\apr{g\vee h}$ satisfies property \ref{pty:card} as
          well. Finally, property \ref{pty:error} is satisfied by
          $\apr{g\vee h}$ due to the definitions of $T_\mcU$ and
          $T_\mcV$.




        \end{itemize}

        \subsubsection*{Error bounds}
        We first show that any $(m,\delta)$-approximator has large error on our
        distributions. \nutan{This sentence is only one part of the story. We could also add to this sentence that we show that the approximator approximates monotone circuits well.}\karteek{The complete story happens in the Techniques section. Then we will add more English here to get the flow going.}
        \begin{lemma}
        \label{lem:ge}
          Let $\tilde{F} = \bigwedge_{i\in S_\mcU} f_i \wedge
          \bigwedge_{j\in S_{\mcV}} g_j$ be any $(m,\delta)$-approximator.
          Then, either $\tilde{F}$ is the constant $0$, or
          \[\Pr_{\Dz}[\tilde{F}=0] \le \frac{2m}{n}  \]
        \end{lemma}
        \begin{proof}
          Suppose $\tilde{F}$ is not the constant $0$. Assume, without
          loss of generality, that $|S_\mcU|>0$. For some $i\in S_\mcU$, let $B := f_i$ be
          a local function of $u_i$. Since $B$ is monotone \nutan{$S_{\mcU}$ is not a set of functions but a set of indices. So, we could say $f_i = B$ for some $i \in S_{\mcU}$ perhaps?}\karteek{Fixed!}
          and not the constant $0$, it must be that $B(1^d)=1$. Recall
          that all vectors not chosen by $\Dz$ in step
          \ref{step:oneD0} are assigned $1^d$. The probability that
          $\Dz$ picks $u_i$ in step \ref{step:oneD0} is $1/n$.  Thus,
          $\Pr_{(\mcU,\mcV)\sim D_0}[B(u_i)=0]\le 1/n$. By a union bound
          over all the local functions in $\tilde{F}$, the
          probability over $\Dz$ that at least one of them is $0$ is
          at most $2m/n$.
        \end{proof}

        We now show that our construction of the $(m,\delta)$-approximator
        incurs very small error at each gate. An observation that is trivial from our 
        construction is that at AND gates (OR gates), our approximator does 
        not incur any new errors on inputs drawn from $\Dz$ ($\Do$ respectively).
        So, we only need to bound the $\Do$-error at AND gates, and $\Dz$ error
        at OR gates.
        \begin{lemma}[Errors at $\wedge$ gates] \label{lem:le1}
          Let $\apr{g}$ and $\apr{h}$ be $(m,\delta)$-approximators
          for the functions computed at gates $g$ and $h$
          respectively. Then we have:
          \[ \Pr_{D_1}[\apr{g}\wedge\apr{h} \neq \apr{g\wedge h}]\le e^{-m\delta} \]

        \end{lemma}
        
        \begin{proof}
        Let
          $\apr{g} = \bigwedge_{i\in S_{\mcU}^g} B_i^{g} \wedge
          \bigwedge_{j\in S_{\mcV}^g} C_j^g$ and
          $\apr{h}=\bigwedge_{i\in S_\mcU^h} B_i^{h} \wedge
          \bigwedge_{j\in S_{\mcV}^h} C_j^{h}$.
          Recall that $\apr{g}\wedge\apr{h}$ always satisfies
          properties \ref{pty:bubble} and \ref{pty:error}.  If
          property \ref{pty:card} was also satisfied, we defined
          $\apr{g\wedge h}= \apr{g}\wedge \apr{h}$. Thus, in this
          case, the lemma is trivially true.

          When $\apr{g}\wedge \apr{h}$ violates property
          \ref{pty:card}, we defined $\apr{g\wedge h}$ to be the
          constant $0$. In this case it suffices to determine the
          probability that $\apr{g}\wedge \apr{h}$ is $1$.  We can
          write
          $\apr{g}\wedge \apr{h} = \bigwedge_{i\in S_{\mcU}^g\cup
            S_{\mcU}^h} P_i\wedge \bigwedge_{j\in S_{\mcV}^g\cup
            S_{\mcV}^h} Q_j$ where $P_i$ and $Q_j$ are defined as follows.\\
          \begin{minipage}{.45\textwidth}
            \[ P_i =
              \begin{cases}
                B_i^g\wedge B_i^h &\text{ for } i\in S_{\mcU}^g\cap S_{\mcU}^h\\
                B_i^g &\text{ for } i\in S_{\mcU}^g\setminus S_{\mcU}^h\\
                B_i^h &\text{ for } i\in S_{\mcU}^h\setminus S_{\mcU}^g\\
              \end{cases}
            \]
          \end{minipage}
                    \begin{minipage}{.45\textwidth}
            \[ Q_j =
              \begin{cases}
                C_j^g\wedge C_j^h &\text{ for } j\in S_{\mcV}^g\cap S_{\mcV}^h\\
                C_j^g &\text{ for } j\in S_{\mcV}^g\setminus S_{\mcV}^h\\
                C_j^h &\text{ for } j\in S_{\mcV}^h\setminus S_{\mcV}^g\\
              \end{cases}
            \]
          \end{minipage}\\
            
            Since property \ref{pty:card}
          is violated, assume without loss of generality that
          $|S_{\mcU}^g\cup S_{\mcU}^h| > m$. Then we have,
          \begin{align*}
    \Pr_{\Do}[\bigwedge_{i\in S_{\mcU}^g\cup S_{\mcU}^h} P_i = 1] &= \prod_i\Pr[P_i = 1]\\
                                                                          &\le \left(1-\delta \right)^{m}\\
            &\le e^{-m\delta}            
          \end{align*}
          where the first equality is because the $P_i$s are on a
          disjoint sets of variables, and thus the events that they
          are $0$ are independent. The inequality in the next line is
          from property \ref{pty:error}.          
        \end{proof}

        We need the following technical lemma before bounding the errors incurred
        at OR gates.
        \begin{lemma}
          \label{lem:mupmuq}
                          Let $B: \{0,1\}^d \rightarrow \{0,1\}$ be
                          any monotone Boolean function. Let
                          $0 < p < q < 1$ and let $\mu_p$ and $\mu_q$
                          be probability distributions where each bit
                          is set to $1$ with probability $p$ and $q$
                          respectively. Then,
			\begin{align*}
                          \Pr_{u \sim \mu_q}[B(u) = 0] \leq \left(\Pr_{u \sim \mu_p}[B(u) = 0]\right)^{\floor{q/p}}
			\end{align*}
		\end{lemma}
		\begin{proof}
                  Let $t = \floor{q/p}$ and let
                  $u_1,u_2,\ldots,u_t \sim \mu_p$ be $d$-bit vectors drawn independently at random 
                  from $\mu_p$. Let $v = u_1 \vee u_2 \vee \ldots \vee u_t$ be the 
                  bit-wise OR of the $u_i$s. We denote the $i$th bit of vector $v$ with
                  $v_i$ and let $q'$ denote $\Pr[v_i = 1]$. We will first observe that
                  $q' \leq q$ and hence the expected number of zeros in
                  a string drawn from the distribution $\mu_{q'}$ will
                  be greater than equal to the expected number of
                  zeros in a string drawn from the distribution
                  $\mu_q$.
                  Using the definition of $v$ and a union bound, we can conclude that $q'\le q$:
			\begin{align*}
                          q' = \Pr_{v  \sim \mu_{q'}}[v_i = 1]  = \Pr_{u_1,u_2,\ldots,u_t \sim \mu_{p}}[(u_{1i} = 1) \vee (u_{2i} = 1) \vee \ldots \vee (u_{ti} = 1)]  \leq tp \leq q
			\end{align*}
             Since $B$ is monotone,
                   $B$ is more likely to output $0$ on a
                  string drawn from $\mu_{q'}$ than on a string
                  drawn from the distribution $\mu_q$. i.e., $
                          \Pr_{v \sim \mu_q}[B(v) = 0] \le \Pr_{v \sim \mu_{q'}}[B(v) = 0]
                                                       = \Pr_{u_1,u_2,\ldots,u_t \sim \mu_{p}}[B(u_1\vee u_2\vee \ldots \vee u_t) = 0]$.
			Combining this with the fact if
                        $B(u_1 \vee u_2 \vee \ldots \vee u_t) = 0$
                        then $\forall$ $i\in [t]$, $B(u_i) =
                        0$, we can conclude:
			\begin{align}
                          \Pr_{u \sim \mu_q}[B(u) = 0] \leq \Pr_{u_1,u_2,\ldots,u_t \sim \mu_{p}}[B(u_1\vee u_2\vee \ldots \vee u_t) = 0]&\le \prod _{i=1}^t\Pr_{u_i \sim \mu_p}[B(u_i) = 0] \notag \\
                                                                                                       & = \left(\Pr_{u_i \sim \mu_p}[B(u_i = 0)]\right)^{\floor{q/p}}  \label{eq:p3}
			\end{align}

                      \end{proof}

        Recall that $p$ is the probability with which $\Do$ sets each bit to
        $1$. We now determine the new errors incurred at $\vee$ gates.
        \begin{lemma}[Errors at $\vee$ gates]  \label{lem:le0}
          Fix parameters $0\le m\le n$ and $0\le \delta\le 1$ such
          that $\delta^{\floor{1/2p}}\le 1/n$.  Let $\apr{g}$ and $\apr{h}$ be
          $(m,\delta)$-approximators for the functions computed at
          gates $g$ and $h$ respectively. Then we have:
          \[ \Pr_{D_0}[\apr{g}\vee \apr{h} \neq \apr{g\vee h}]\in O\left(\frac{m^2}{n^2}\right) \]
        \end{lemma}
        \begin{proof}
          Let
          $\apr{g} = \bigwedge_{i\in S_{\mcU}^g} B_i^{g} \wedge
          \bigwedge_{j\in S_{\mcV}^g} C_j^g$ and
          $\apr{h}=\bigwedge_{i\in S_\mcU^h} B_i^{h} \wedge
          \bigwedge_{j\in S_{\mcV}^h} C_j^{h}$ be
          $(m,\delta)$-approximators for $g$ and $h$
          respectively. Recall that $\apr{g}\vee \apr{h}$ looks like
          equation \ref{eq:apprvee}. When constructing
          $\apr{g\vee h}$, we removed functions of two kinds
          from equation \ref{eq:apprvee}:
          \begin{enumerate}
          \item Functions that were not local functions. These come in two types:
            \begin{enumerate}
            \item $(B_i^g\vee B_j^h)$ for some $i\neq j$. Recall that
              in the random process that defines $\Dz$, all vectors in
              $\mcU$ except one (randomly chosen) vector are made
              $1^d$. Thus, it is always true that either $u_i$ or $u_j$ is assigned
              $1^d$. This means either $B_i^g$ or $B_j^h$ must be $1$
              on every input drawn from $\Dz$. Hence there is no error
              when not including such functions in $\apr{g\vee h}$.
            \item\label{itm:oneB} $(B_i^g\vee C_j^h)$. If either of $i$ or $j$ were
              not the chosen indices in step \ref{step:oneD0} of
              $D_0$, then similar to the previous case, either $B_i^g$
              or $C_j^h$ will be $1$.  The probability that $i$ as
              well as $j$ were the chosen indices in step
              \ref{step:oneD0} of $\Dz$ is $1/n^2$. There are at most
              $m^2$ such functions, and hence by a union bound, the
              error incurred is at most $m^2/n^2$.
             
            \end{enumerate}
          \item Functions that were local functions, but did not
            satisfy property \ref{pty:error}. Suppose
            $B_i^g\vee B_i^h$ did not satisfy property \ref{pty:error}.
            We will determine the probability that $B_i^g\vee B_i^h$ equals
            $0$.

            Observe that if $i$ was the first index chosen by the
            distribution $\Dz$ in step \ref{step:oneD0}, then each bit
            of the vector $u_i$ is chosen uniformly and independently
            at random. Let $B=B_i^g\vee B_i^h$. Since $B$ does not
            satisfy property \ref{pty:error}, we have
            $\Pr_{u_i\sim \Do}[B(u_i)=0] < \delta$. Applying lemma
            \ref{lem:mupmuq} to $B$ with $q = 1/2$, we obtain:
            
            \[ \Pr_{u_i\sim \mu_{1/2}}[B(u_i) = 0] \le \left(\Pr_{x\sim \Do}[B(u_i)=0]\right)^{\floor{1/2p}} < \delta^{\floor{1/2p}} \]

              Thus, $\Pr_{(\mcU,\mcV)\sim \Dz}[B(u_i) = 0] = \Pr[i\text{ was chosen}]\cdot\Pr_{u_i\sim \mu_{1/2}}[B(u_i) = 0] < \delta^{\floor{1/2p}}/n$.
              By a union bound over at most $m^2$ such functions, the
              total error incurred is at most
              $\delta^{\floor{1/2p}}m^2/n$. Since
              $\delta^{\floor{1/2p}} \le 1/n$, the lemma follows.
            
          \end{enumerate}
          
        \end{proof}

            We can now show the main theorem by combining the error bounds established above.
                \maintheorem*
                \begin{proof}
              	Set $p = \epsilon/4$, $c = 100/\epsilon^2$,
                $d = c\log n$, $m = n^{\epsilon}$ and
                $\delta = \frac{1}{n^{\epsilon/2}}$.  Let $C$ be any monotone circuit computing $\Intnd$ using $s$ many
                gates.  Let $\tilde{F}$ be the $(m,\delta)$-approximator for $C$ constructed 
                as described in Section \ref{sec:construction}. As mentioned before, the 
                approximation at AND gates incurs new errors only over the distribution 
                $\Do$. Similarly, when approximating OR gates, we only incur new 
                errors on inputs drawn from $\Dz$. Using error bounds from Lemma \ref{lem:le0},
                Lemma \ref{lem:le1}, and a union bound over the $s$ gates of $C$, we have:
                \begin{align}
                    \Pr_{\Dz}[\tilde{F}\neq C] &\le s\cdot O\left(\frac{m^2}{n^2}\right) \label{eq:local0error} \\
                    \Pr_{\Do}[\tilde{F}\neq C] &\le s e^{-m\delta} \label{eq:local1error}
                \end{align}
                
                Now we have two cases with respect to whether $\tilde{F}$ is the constant $0$
                or a non-constant function.
              
              	\begin{itemize}
                \item \textbf{Case 1:} $\tilde{F}$ is the
                  constant 0. In this case, $\tilde{F}$ is wrong on all $1$-inputs of $\Int$ 
                  drawn from the distribution $\Do$. So, we have:
                  \begin{align*}
                    \Pr_{\mcU,\mcV \sim \Do}[C(\mcU,\mcV) \neq \tilde{F}(\mcU,\mcV)] &= \Pr_{\mcU,\mcV \sim \Do}
                    [C(\mcU,\mcV) = 1 \text{ and } \tilde{F}(\mcU,\mcV)=0]\\ 
                    &= \Pr_{\mcU,\mcV \sim \Do}[\Intnd(\mcU,\mcV) = 1]\\
                    &\ge 1 - n^2 e^{-p^2d}\in \Omega(1)
                  \end{align*}
                                
                  where the first two equalities are because $\tilde{F}$ is assumed to be $0$.  The final inequality is from Lemma~\ref{lem:pr1}. \nutan{Better to use capital letters for "Lemma" etc.?}\karteek{Done} Combining this with 
                  Equation~\ref{eq:local1error}, we obtain:
                  \begin{align*}
                    \Omega(1) \le \Pr_{\Do}[C \neq \tilde{F}] \le s e^{-m\delta}.
                  \end{align*}
                  
                  This gives a lower bound of $s = \Omega(e^{m\delta})$.

                \item \textbf{Case 2: } $\tilde{F}$ accepts at least one input. Combining the error bound from 
                Lemma~\ref{lem:ge} with Equation~\ref{eq:local0error}, we get:
                \begin{align*}
                   1 - \frac{2m}{n}  \le \Pr_{\mcU,\mcV \sim \Dz}[C(\mcU,\mcV)\neq \tilde{F}(\mcU,\mcV)] \le s\cdot O\left(\frac{m^2}{n^2}\right)
                \end{align*}
              	This allows us to conclude $s \ge \Omega\left(\frac{n^2}{m^2}\right)\left(1-\frac{2m}{n}\right) = \Omega(n^{2(1-\epsilon)})$.
                \end{itemize}
                Taking into account both of the above two cases, we get $s \geq \min({\Omega(n^{2(1-\epsilon)}),e^{n^{\epsilon/2}}})$. Thus, for large $n$, we have $s \geq \Omega(n^{2(1-\epsilon)})$.
              \end{proof}
                \nutan{Removed the remark about kOV.}

              \nutan{Todo: add Jukna criterion related comment here.}\karteek{Done}

    \section{Formula and Branching Program Lower Bounds}\label{sec:Nechiporuk}
    \srikanth{Replace $\twoint$ by $\Intnd$.} \karteek{Done}
    In this section we show lower bounds on the size of (not necessarily monotone) formulas and branching programs computing $\Intnd$. We use the classic technique of Nechiporuk \cite{juknabook,nechiporuk1966} \karteek{TODO (Tameem?): Cite Nechiporuk's result or Jukna's book. (This is done!)} that relates the formula and branching program complexity of a function with the number of sub-functions with respect to disjoint subsets of variables.

    \karteek{TODO: (Prelims) Introduce/refer to standard text for definition of formulas and BPs. Done.}
    \srikanth{Please note the footnote that I have added. If the theorem is changed to this more general version, this footnote can be removed.}\karteek{Decided not to rewrite stuff at this point.}
    Throughout this section, we shall use $L(f)$ to denote the number of leaves in a smallest
                Boolean formula (with arbitrary binary gates)\footnote{A similar proof also works for the case where the fan-in of the formulas are bounded by some absolute constant $c_0 \geq 2$ and each gate is labelled by an arbitrary Boolean function on at most $c_0$ variables.} computing $f$, and BP$(f)$ to denote the number of
                nodes in a smallest branching program computing $f$. We first state Nechiporuk's theorem.
	\begin{theorem}[\cite{juknabook,nechiporuk1966}] \label{thm:Nechiporuk} Let $f$ be a Boolean function on a
          set of variables $X$.  Let
          $Y_1,Y_2,\ldots,Y_m$ be disjoint subsets of $X$, and let
          $s_i$ be the number of distinct sub-functions of $f$ on
          $Y_i$. Then, there exists a constant
          $\epsilon > 0$ such that
		\begin{align}
                  L(f) &\geq \frac{1}{4} \sum_{i = 1}^m \log s_i \label{eq:1}\\
                  \text{BP}(f) &\geq \epsilon \sum_{i = 1}^m \frac{\log s_i}{\log \log s_i} \label{eq:2}
		\end{align}
        \end{theorem}
        
        \begin{theorem}\label{thm:formulaBPLowerBound}
          We obtain lower bounds in two separate ranges of $d$:
          \begin{enumerate}
          \item For $d \ge c\log n$ for any $c > 1$,
            \begin{align*}
              L(\Intnd) &\ge \Omega(n^2d)\\
              \text{BP}(\Intnd) &\ge \Omega\left(\frac{n^2d}{\log(nd)}\right)
            \end{align*}
            
          \item For $d \le \log n$,
            \begin{align*}
              L(\Intnd)&\ge \Omega(n2^d/\sqrt{d})\\
              \text{BP}(\Intnd)&\ge \Omega\left(\frac{n 2^d}{\sqrt{d}\log(2nd)}\right)
            \end{align*}
          \end{enumerate}
	\end{theorem}

	To use Theorem \ref{thm:Nechiporuk}, we first bound the number of distinct sub-functions of $\Intnd$ with respect to the variables in any single vector.
    
	\begin{lemma}\label{lem:subfunclarged}
          Let $d \geq c\log n$ for any $c > 1$. Consider an instance
          of $\Int(\mcA,\mcB)$. Let $a_i$ denote the $i$th
          vector in $A$ and $b_i$ denote the $i$th vector in $B$. Let
          $s_i$ ($t_i$) denote the number of distinct sub-functions of
          $\Int(\mcA,\mcB)$ with respect to $a_i$ ($b_i$
          respectively). Then the following inequalities hold:
          \[ \left(\frac{2^{d-1}}{n\sqrt{d}}\right)^n \le  s_i, t_i \le 2^{2nd} \]
	\end{lemma}
	\begin{proof}
          We will show the statement for the number of distinct
          sub-functions $s_1$ of $\Int$ with respect to $a_1$. A
          similar argument holds for each $s_i$ and $t_i$.

          Suppose we set all the vectors in $\mcA$ other than
          $a_1$ to $\vec{1}$, and set the vectors in $\mcB$ to non-zero vectors. \srikanth{Here, setting $a_2 = \vec{1}$ doesn't automatically set the corresponding DNF to $1$. That happens only because all the $b_j$s are set to non-zero vectors, which happens later. So we should rephrase?} \karteek{Rephrased slightly!!}\srikanth{Perfect} Then each of $a_2,\ldots, a_n$ will trivially intersect each vector in $\mcB$ leaving us with the following sub-function on $a_1$: 
          $\bigwedge_{i = 1}^n \bigvee_{j = 1}^d (a_{1j} \wedge b_{ij})$.  
          Thus, setting $b_1,\ldots, b_n$ to non-zero vectors
          results in a monotone CNF on the variables of $a_1$.
          
          \textbf{Lower bound.}
          Let $\Hd$ be the set of all $d$-dim Boolean
          vectors having exactly $d/2$ many ones. We claim that if we
          set each $b_i \in \mcB$ to a distinct vector from $\Hd$, then
          each way of setting the $b_i$s results in distinct
          sub-functions. Observe that if the $b_i$s have exactly $d/2$
          many $1$s, then each clause in the CNF above has exactly
          $d/2$ positive literals, and these clauses encode exactly the maxterms of the corresponding sub-function.\srikanth{Rephrased slightly here.} 
          Hence each such distinct CNF is a distinct
          sub-function of $\Intnd$ with respect to $a_1$. Therefore
          the number of sub-functions is at least the number of ways
          to choose $n$ vectors from $\Hd$.

          Since $d = c\log n$ and $c > 1$, for large $n$, we have:
          \[ |\Hd| = \binom{d}{d/2}\geq \frac{2^{d}}{2\sqrt{d/2}} = \frac{n^c}{\sqrt{2c\log n}} > n \]

          Thus, the number of distinct subfunctions is:
          \[ s_1\ge \binom{|\Hd|}{n} \geq \binom{2^{d}/2\sqrt{d/2}}{n} \geq
          \left(\frac{2^{d}}{ n \sqrt{2d}}\right)^n \geq \left(\frac{2^{d-1}}{n \sqrt{d}}\right)^n  \]

        \textbf{Upper bound.}
        To upper bound the number of distinct sub-functions on $a_1$,
        it suffices to bound the number of different settings to the
        vectors $a_2,\ldots, a_n, b_1,\ldots, b_n$. This is at most
        $2^{(2n-1)d} \le 2^{2nd}$.
        \end{proof}
    \karteek{Should we merge these two Lemmas?}
        \begin{lemma}
          \label{lem:subfuncsmalld}
          Let $d \le \log n$. The following inequalities hold for the
          number of distinct sub-functions $s_i$ of $\Intnd$ with
          respect to any $a_i$ or $b_i$:
          \[ 2^{\binom{d}{d/2}} \le s_i \le 2^{2nd} \]
        \end{lemma}
        \begin{proof}
          As described in proof of Lemma \ref{lem:subfunclarged}, we
          can set the $b_i$s to obtain a monotone CNF with clauses on
          variables of $a_1$. Observe that since $d\le \log n$, there
          are at most $2^d \le n$ many possible monotone
          clauses. Thus, we have sufficiently many $b_i$ to encode any
          monotone CNF on the variables of $a_i$. Hence, \emph{every}
          monotone function on the $d$ variables of $a_1$ are
          sub-functions of $\Intnd$ with respect to $a_1$. The number
          of monotone functions on $d$ variables is at least the
          number of slice functions on the $d/2$'th slice. This is
          $2^{\binom{d}{d/2}}$.

          The upper bound is exactly as in proof of Lemma \ref{lem:subfunclarged}.
        \end{proof}

        \begin{proof}[Proof of Theorem \ref{thm:formulaBPLowerBound}]
          Partition the $2nd$ variables into parts that correspond to
          each vector. i.e., for all $i\in [n]$, the part $Y_i = a_i$ and for $j\in \{n+1,\ldots, 2n\}$,
          the part $Y_j = b_j$.

          The case of $d\ge c\log n$ for $c>1$ is as follows. Combining Lemma \ref{lem:subfunclarged} and Theorem~\ref{thm:Nechiporuk}, we obtain the following.
	\begin{align*}
		L(\Intnd) &\geq \frac{1}{4} \sum_{i = 1}^{2n} \log \left(\frac{2^{d-1}}{n\sqrt{d}}\right)^n = \frac{n^2}{2}(d-1 - (\log n+1/2\log d)) \in \Omega(n^2d)\\ 
        BP(\Intnd) &\geq \epsilon \sum_{i = 1}^{2n}\frac{\log\left(\frac{2^{d-1}}{n\sqrt{d}}\right)^n}{\log\log 2^{2nd}}
              = \epsilon\, 2n \frac{n(d-1-(\log n + 1/2\log d))}{\log 2nd}\in \Omega\left(\frac{n^2d}{\log(nd)}\right)
      \end{align*}

        For the case of $d\le \log n$, we use Lemma \ref{lem:subfuncsmalld} in Theorem \ref{thm:Nechiporuk}:        \begin{align*}
          L(\Intnd) &\geq \frac{1}{4} \sum_{i = 1}^{2n} \log 2^{\binom{d}{d/2}} = \frac{n}{2}\binom{d}{d/2} \ge \frac{n2^d}{4\sqrt{d/2}} \in \Omega\left(\frac{n2^d}{\sqrt{d}}\right)\\
            BP(\Intnd) &\geq \epsilon \sum_{i = 1}^{2n} \frac{\log 2^{\binom{d}{d/2}}}{\log\log 2^{2nd}}\in \Omega\left(\frac{n 2^d}{\sqrt{d}\log(2nd)}\right)
      \end{align*}
      
      \end{proof}

        \begin{remark} 
            We note that Kane and Williams \cite{kane2019orthogonal} established a  lower bound of $L(\Intnd) \geq \Omega\left(\min(\frac{n^2}{\log d},\frac{n2^d}{\sqrt{d}\log d})\right)$  for formulas computing $\Intnd$ and a lower bound of $BP(\Intnd) \geq \Omega\left(\min(\frac{n^2}{\log d \log(nd)},\frac{n2^d}{\sqrt{d}\log d \log(nd)}) \right)$ for branching programs computing $\Intnd$. For $c > 1$ and $d \geq c\log n$ our result improves this bound for both formulas and branching programs by a factor of $d \log d$. When $d \leq \log n$ the improvement is by a factor of $\log d$ in both the models. 
            \karteek{Tameem, complete this for $d\le \log n$, and compare with our theorem}
        \end{remark}
        \tameem{Done.}
              \section{Monotone Upper Bounds for $\Intnd$}\srikanth{Remove this section as it talks about $k$-OV?}\karteek{Since the main theorem is a monotone lower bound, I think it makes sense to keep a monotone upper bound. But we could downgrade this from kInt to 2Int?}
              \srikanth{Sounds fine to downgrade to 2Int.}
              \karteek{"Downgrade" complete!}
In
this section, we show a construction of monotone circuits to compute
$\Int$.  We first need the following characterization of $0$-inputs of
$\Intnd$.
              \label{sec:upperbounds}
              \begin{lemma} \label{lem:zeroinput} Let
                $x = (\mcA,\mcB)$ be an input of $\Intnd$. Then
                $\Intnd(x) = 0$ if and only if there exists an ordered
                partition $\mathcal{P} = (P_1, P_2)$ of $[d]$ into two
                parts such that there exists $a \in \mcA, b\in \mcB$
                satisfying $a \cap P_1 = \emptyset$ and
                $b\cap P_2 = \emptyset$
              \end{lemma}
          \begin{proof}
            ``$\implies$'': Let $x = (\mcA,\mcB)$ be a $0$-input of
            $\Intnd$. From Definition \ref{def:Int} there must exist
            $a \in \mcA,b\in\mcB$ such that $a\cap b =
            \emptyset$. Hence for all $j \in [d]$, either $a_j = 0$ or
            $b_j = 0$ or both.  For every $j \in [d]$, we include $j$ in the
            set $P_1$ if $a_j = 0$, and in $P_2$ if
            $b_j=0$. If both $a_j$ and $b_j$ are $0$, then we include
            $j$ in only $P_1$. Therefore each $j \in [d]$ belongs to
            exactly one of $P_1$ and $P_2$. This results is an ordered
            partition of $[d]$ into two parts namely $P_1$ and
            $P_2$. It is now straightforward to observe that
            $\mathcal{P} = (P_1,P_2)$ satisfies the claim.
            
            ``$\impliedby$'': Let $x = (\mcA,\mcB)$ be an input of
            $\Intnd$. Suppose $\mathcal{P} = (P_1,P_2)$ is an ordered
            partition of $[d]$ into two parts such that
            $a \cap P_1 = \emptyset$ and $b\cap P_2 =
            \emptyset$. Then, consider any $j \in [d]$.  Since
            $\mathcal{P}$ is a partition of $[d]$, either $j\in P_1$
            or $j\in P_2$. If $j\in P_1$, then $a_j=0$, else $j\in P_1$ and $b_j=0$.
            Hence $x$ is a  $0$-input of $\Intnd$.
          \end{proof}
	\begin{theorem} \label{thm:ub} The function
          $\Intnd$ can be computed by a monotone circuit of size at
          most $O(\min\{2^dn d, n^{2}d\})$.
	\end{theorem}
 
		\begin{proof}
                  We will construct two circuits $C_1$, and $C_2$,
                  computing $\Intnd$ with sizes $O(2^dnd)$, and
                  $O(n^2d)$ respectively.
		
		The circuit $C_1$ is a brute-force through all
                possible ordered partitions to check if any of them
                satisfy Lemma \ref{lem:zeroinput}. Let $\mathcal{S}$
                be the set of all ordered partitions of $[d]$ into two
                parts. The following circuit outputs $0$ if for some
                fixed ordered partition
                $\mathcal{P} = (P_1,P_2)$, there exists
                $a \in \mcA, b\in \mcB$ such that
                $a \cap P_1 = \emptyset$ and $b\cap P_2 = \emptyset$.
		\begin{align*}
                  C_{\mathcal{P}} = \left(\bigwedge_{a \in \mcA} \bigvee_{j \in P_1} a_{j}\right)\bigvee \left(\bigwedge_{b \in \mcB} \bigvee_{j \in P_2} b_{j}\right)
		\end{align*}
		Finally to check if there exists an ordered partition
                $\mathcal{P}$ such that $C_{\mathcal{P}}$ outputs $0$
                we $\AND$ over all ordered partitions of $[d]$ into
                two parts.
		\begin{align*}
			C_1 = \bigwedge_{\mathcal{P} \in \mathcal{S}} C_{\mathcal{P}}
		\end{align*}
                Here each $C_{\mathcal{P}}$ is of size at most $O(nd)$
                and $|S| = O(2^d)$. Therefore size of $C$ is at most
                $O(2^dnd)$.
	
	The circuit $C_2$ is a brute-force search to match Definition
        \ref{def:Int}:
	\begin{align*}
          C_2 = \bigwedge_{a\in \mcA, b\in \mcB}\:\: \bigvee_{j \in [d]}\:\: a_j\wedge b_j
	\end{align*}  
	Since, each $|A|=|B|=n$, the size of $C$ is $O(n^2d)$.
	\end{proof}

    \section{Jukna's criterion for $\Intnd$}
\label{sec:Jukna}
    
            Jukna \cite{juknabook} (Section 9.4) gives a general criterion to obtain lower bounds for monotone circuits based on the monotone switching lemma of Berg and Ulfberg~\cite{BergUlfberg}. This criterion is as follows.

            \begin{theorem}
                \label{thm:Jukna}
                Assume that $f$ is a monotone Boolean function on $n$ variables and has monotone circuits of size $t$. Then, for any $2\leq r, s \leq n$, there is an exact $s$-CNF $\varphi$, an exact $r$-DNF $\psi$ and a subset $X'$ of the variables of size at most $s-1$ such that
                \begin{itemize}
                    \item $\varphi$ has at most $t\cdot (r-1)^s$ clauses and $\psi$ has at most $t\cdot (s-1)^r$ terms,\footnote{An \emph{Exact $s$-CNF} is one where every clause has exactly $s$ distinct variables. Exact $r$-DNFs are defined similarly.} and
                    \item $\varphi\leq f$ or $f\leq \psi\vee \bigvee_{x\in X'} x$.
                \end{itemize}
            \end{theorem}

            To use the above for a lower bound for $\Intnd$, we need to choose a suitable $r,s$ so that the conclusion is not true for $\Intnd.$ Unfortunately, this is not possible as long as $c_0\log n \leq d = n^{o(1)}$ for a large enough absolute constant $c_0.$ We argue this below.

            Our argument rests on the following  claims.

            \begin{claim}
            \label{clm:rsmall}
                For any $2\leq r \leq n,$ there is an exact $r$-DNF $\psi_0$ with at most $d^r$ terms such that $\Intnd \leq \psi_0.$ \nutan{Do we need to explain this $\leq$?}
            \end{claim}

             \begin{claim}
            \label{clm:ssmall}
                For any $s <  d/2,$ there is an exact $s$-CNF $\varphi_1$ with at most $nd\cdot 2^{O(s)}$ clauses such that $\varphi_1\leq \Intnd.$
            \end{claim}

            \begin{claim}
            \label{clm:slarge}
                For any $n,d$, $\Intnd$ can be written as an exact $d$-CNF $\varphi_0$ with at most $n^2\cdot 2^{O(d)}$ clauses.
            \end{claim}

        Given the above claims, we can prove that for any choice of $r,s\in \{2,\ldots, 2nd\}$, $\Intnd$ satisfies the conclusion of Theorem~\ref{thm:Jukna} for $t = n^{1+o(1)}$, and hence Theorem~\ref{thm:Jukna} cannot be used to prove a near-quadratic lower bound. This is done as follows.
        
        \begin{itemize}
        \item If $s \geq d+1$, we note that $\Intnd \leq \bigvee_{i=1}^d a_{1i}$ and hence there is a subset $X'$ of at most $s-1$ variables such that $\Intnd \leq \bigvee_{x\in X'}x.$ For the remainder of the argument, assume that $s \leq d.$
        \item If $r\leq r_0$ for a large enough constant $r_0$ (chosen below), then we use 
        Claim~\ref{clm:rsmall}, which gives us an exact $r$-DNF $\psi_0$ such that $\Intnd\leq \psi_0$. Note that $d^r\leq d^{r_0}\leq n^{o(1)}$ as $d\leq n^{o(1)}$ and hence we are done. For the remainder of the argument, assume that $r > r_0.$
        \item If $s \in [d/2,d],$ then we consider the instance of $\mathsf{Int}_{n,s}$ obtained by setting the last $d-s$ entries in each vector in the input lists $A,B$ to $0.$ Note that $\mathsf{Int}_{n,s}\leq \Intnd$ and by Claim~\ref{clm:slarge}, we can write it as an exact $s$-CNF $\varphi$ with at most $n^2\cdot 2^{O(s)}$ clauses. As $s \geq d/2$ and $d\geq c_0\log n$, the number of clauses is at most $(r-1)^s$ as long as $r_0$ and $c_0$ are large enough. Hence we are done in this case. 
        \item If $s < d/2,$ then we use Claim~\ref{clm:ssmall} to directly obtain an exact $s$-CNF $\varphi_1\leq \Intnd$ with at most $nd\cdot 2^{O(s)}\leq n^{1+o(1)}\cdot (r-1)^s$ as $d\leq n^{o(1)}$ and for a large enough constant $r_0.$ This finishes the argument.
        \end{itemize}
        
        \subsection{Proof of Claims}

        In this section, we sketch the proofs of the above claims, all of which are straightforward.

        \begin{proof}[Proof of Claim~\ref{clm:rsmall}]
            This can be done by writing down an exact $r$-DNF $\psi_0$ that is zero exactly when one of the vectors $a_1,\ldots, a_r$ in the list $A$ to $0$. Note that this implies that $\Intnd$ is also $0$ and hence $\Intnd\leq \psi_0.$ We can write down such a $\psi_0$ by
            \[
            \psi_0 = \bigvee_{i_1,\ldots,i_r\in [d]} a_{1,i_1}\wedge a_{2,i_2}\cdots \wedge a_{r,i_r}.
            \]
            This clearly has $d^r$ terms.
        \end{proof}

        \begin{proof}[Proof of Claim~\ref{clm:slarge}]
            This can be done by writing out the monotone depth-$3$ circuit $C_2$ from the previous section as an exact $d$-CNF. More precisely, we have
            \[
            \varphi_0(A,B) = \bigwedge_{i,j\in [n]}\bigwedge_{S\subseteq [d]}(\bigvee_{p\in S} a_{i,p}\vee \bigvee_{q\not\in S} b_{j,q}).
            \]
            This has the required properties.
        \end{proof}

        \begin{proof}[Proof of Claim~\ref{clm:ssmall}]
            We show that for any set $Y = \{y_1,\ldots,y_d\}$ of $d$ Boolean variables, there is an exact $s$-CNF $\varphi'$ with $d\cdot 2^{O(s)}$ clauses such that 
            \begin{equation}
                \label{eq:randomcnf}
                \varphi'(y) = 1\Longrightarrow |y| > d/2.
            \end{equation}
            Given this, we can construct $\varphi_0$ as follows
            \[
            \varphi_1 = \bigwedge_{i=1}^n \varphi'(a_i) \wedge \bigwedge_{j=1}^n \varphi'(b_j).
            \]
            If $\varphi_0(A,B) = 1$, then each vector has Hamming weight greater than $d/2$, implying that all pairs intersect and hence $\varphi_0 \leq \Intnd$. Moreover, the total size is $2nd\cdot 2^{O(s)} = nd\cdot 2^{O(s)}.$

            It remains to construct $\varphi'.$ This is done by a standard random argument. Pick $T$ clauses $C_1,\ldots, C_T$ of size exactly $s$ independently and uniformly at random from among the variables $Y$ and set $\varphi' = C_1\wedge C_2 \wedge \cdots \wedge C_T.$ For any setting $y$ to the variables of weight at most $d/2$, the probability that a single clause is not satisfied is at least
            \[
            \frac{\binom{d/2}{s}}{\binom{d}{s}} \geq \frac{(d/2s)^s}{(ed/s)^s} = \frac{1}{2^{O(s)}}
            \]
            where the first inequality uses Stirling approximations. Thus, the probability that $\varphi'(y) = 1$ is at most
            \[
            \left(1-\frac{1}{2^{O(s)}}\right)^T < \frac{1}{2^d}
            \]
            for $T = d\cdot 2^{O(s)}.$ By the probabilistic method, this implies that there is a $\varphi'$ with the required properties. This finishes the argument.
        \end{proof}

        \bibliographystyle{plainurl}
	{\small
		\bibliography{references}
        }
\end{document}